\documentclass[11pt,a4paper]{article}

\newcommand{\Style}[2]{#2}

\newcommand{\QED}{}

\usepackage{amsmath}
\usepackage{enumerate}
\usepackage{url}

\newcommand{\True}{\mathcal{T}}
\newcommand{\False}{\mathcal{F}}
\newcommand{\IEval}[2]{\text{eval}(#1, #2)}
\newcommand{\IWidth}[1]{w(#1)}
\newcommand{\Sign}[1]{\mbox{sgn}(#1)}
\newcommand{\Iv}[1]{\mathbf{#1}}
\newcommand{\RR}{\mathbb{R}}
\newcommand{\Opt}[2]{#1 \;\;\text{subject to } #2}

\newcommand{\Long}[1]{}

\newtheorem{lemma}{Lemma}
\newtheorem{definition}{Definition}
\newtheorem{theorem}{Theorem}
\Style{}{\newenvironment{proof}{\noindent{\bf Proof.}\hspace{0.05in}}{\par\mbox{}\par}}

\usepackage{amssymb}
\usepackage{color}
\usepackage{ntabbing}

\begin{document}

\title{Efficient Solution of a Class of Quantified Constraints with Quantifier
  Prefix Exists-Forall}
\Style{
\author{Milan Hlad{\'\i}k}
\address{Charles University, Prague, Czech Republic}
\email{hladik@kam.mff.cuni.cz}

\author{Stefan Ratschan}
\address{Institute of Computer Science, Academy of Sciences of the Czech
  Republic, Prague, Czech Republic\\
  ORCID: 0000-0003-1710-1513}
\email{stefan.ratschan@cs.cas.cz}
\thanks{Stefan Ratschan's work was supported by the Czech Science Foundation
  (GA{\v C}R) grant number P202/12/J060 with institutional support
  RVO:67985807. Milan Hladik's work was supported by the Czech Science
  foundation (GA{\v C}R) grant number P402-13-10660S.}}{
\author{Milan Hlad{\'\i}k and Stefan Ratschan}
\maketitle}

\begin{abstract}
  In various applications the search for certificates for certain properties (e.g., stability of dynamical systems, program
  termination) can be formulated as a quantified constraint solving problem with
  quantifier prefix exists-forall. In this paper, we present an algorithm for
  solving a certain class of such problems based on interval techniques in
  combination with conservative linear programming approximation. In comparison
  with previous work, the method is more general---allowing general Boolean
  structure in the input constraint, and more efficient---using splitting
  heuristics that learn from the success of previous linear programming approximations.
\end{abstract}

\maketitle
\section{Problem Description}

We study the problem of finding $x_1,\dots,x_r$ such that
\[\bigwedge_{i=1}^n \forall y_1,\dots, y_s\!\in\!
\Iv{B}_i\;\; \phi_i(x_1,\dots,x_r, y_1,\dots y_s)\] where each $\Iv{B}_i$ is a
box (i.e., Cartesian product of closed intervals) in
$\RR^s$ and each of the $\phi_1,\dots,\phi_n$ is a Boolean combination of
inequalities where for each $i\in \{1,\dots,n\}$ only one of those inequalities contains the variables
$x_1,\dots,x_r$ and this one inequality contains those variables only linearly%
\Long{\footnote{actually the algorithm can be easily extended to the case where the
  other disjunctions contain $x_1,\dots, x_r$, but probably it will not be
  (quasi)complete.}}. If no such $x_1,\dots,x_r$ exist, we want to detect
this. Here is an illustrating example: 
\[
\begin{array}{c}
\forall y_1\in [0,1], y_2\in [-1,1] \;\;\; [y_1\geq y_2 \vee x_1\sin(y_1)y_2 + x_2 y_1^2y_2\leq 0]\\
\wedge\\
\forall y_1\in [0,1], y_2\in [-1,1] \;\;\; [y_1< y_2 \vee x_1\cos(y_1)y_2 + x_2 y_1y_2^2\leq 0].
\end{array}
\]

We also study the extension of this problem to the case where the conjunction may---in
addition to constraints of the form \[\forall y_1,\dots, y_s\!\in\!
\Iv{B}_i\;\; \phi_i(x_1,\dots,x_r, y_1,\dots y_s)\]---also contain linear equalities
in the variables $x_1,\dots, x_r$ (the
equalities can be viewed as a conjunction of inequalities, but this violates the
condition that only one inequality contains $x_1,\dots, x_r$).

In an earlier paper~\cite{Ratschan:10}, we showed how to solve a special case, with restricted
Boolean structure. The contributions of this paper are: 
\begin{itemize}
\item The extension of the
approach to arbitrary Boolean structure. 
\item The design of splitting heuristics that improves the performance of the
  algorithm by orders of magnitude.
\end{itemize}

Constraints of this type occur in various applications. Especially, they are
useful in finding certificates for certain global system properties. For
example, a Lyapunov function~\cite{Khalil:02} represents a certificate for the stability of dynamical
systems. In the case of global stability, such a function has to fulfill certain
properties in the whole state space except for the original/equilibrium. After
using an ansatz (often also called template) of the Lyapunov function as a polynomial with parametric
coefficients, one can find the Lyapunov function by solving a universally
quantified problem for those parameters. The fact that polynomials are
linear in their coefficients corresponds to linearity of our variables
$x_1,\dots, x_r$. A similar situation occurs, for example in termination analysis of computer
programs~\cite{Podelski:04,Cousot:05}, or in the termination analysis of term-rewrite systems~\cite{Dershowitz:87}.

However, usually further work is usually necessary to apply the method studied
in this paper to such problems: In the case of Lyapunov functions, one has to
exempt one single point (the equilibrium) from the property which cannot be
directly expressed by the boxes $\Iv{B}_1,\dots,\Iv{B}_n$. In the case of termination
analysis, such constraints have to be solved for the
whole real space instead of boxes $\Iv{B}_1,\dots,\Iv{B}_n$. In an earlier paper, we
solved this problem for Lyapunov
functions~\cite{Ratschan:10,Giesl:14}. However, in other areas this is an open
area for further research.

In the polynomial case, constraints over the domain of real numbers with quantifiers can
always be solved due to Tarski's classical result that the first-order theory of
the real numbers allows quantifier elimination~\cite{Tarski:51}. The area of
computer algebra has developed impressive software packages for quantifier
elimination~\cite{Brown:04,Dolzmann:97b}. However, those still have problems with scalability in the number of
involved variables and, in general, they cannot solve non-polynomial problems.
This was the motivation for several approaches to use interval based techniques
for such problems, for special cases~\cite{Jaulin:96,Benhamou:00,Goldsztejn:09}
and for arbitrary quantifier structure~\cite{Ratschan:02f}.

Alternative approaches for finding Lyapunov function certificates are based on
techniques from real algebraic
geometry~\cite{Prajna:05b,Sankaranarayanan:13}. However, those methods cannot
solve general constraints of the form discussed in this paper (no general
Boolean structure, no non-polynomial constraints).

In termination analysis of term-rewrite systems constraints with
quantifier-prefix $\exists\forall$ are usually solved by first eliminating the universally
quantified variables using conservative approximation~\cite{Hong:98,Lucas:07},
and then solving the remaining, existentially quantified problem. Again, this
technique cannot solve general constraints of the form discussed in this paper (no general
Boolean structure, no non-polynomial constraints).

The structure of the paper is a as follows: In the next section, we will
introduce the basic algorithm for solving the constraints. In
Section~\ref{sec:heuristics} we will introduce splitting heuristics for the
algorithm. In Section~\ref{sec:equality-constraints} we will extend the
algorithm with equality constraints. In Section~\ref{sec:convergence} we will 
prove termination of the algorithm for all non-degenerate cases. In
Section~\ref{sec:variable_splitting} we discuss how, for a given box to split,
choose the variable of that box to split. In Section~\ref{sec:experiments} we
will provide the results of computational experiments with the algorithm. And in
Section~\ref{sec:conclusion} we will conclude the paper.

Throughout the paper, boldface variables denote objects that are intervals
or contain intervals (e.g., interval vectors or matrices).

\section{Basic Algorithm}

We use the following algorithm (which generalizes an
algorithm~\cite{Ratschan:10} that solves constraints of a more specific form
arising in the analysis of ordinary differential equations):
\begin{enumerate}
\item\label{item:iveval} For each $i\in \{ 1,\dots,n\}$, substitute the intervals of $\Iv{B}_i$
  corresponding to $y_1,\dots,y_s$ into $\phi_i$, and evaluate using interval
  arithmetic. As a result, all the inequalities that do not contain
  $x_1,\dots,x_r$ are simplified to an inequality of the form $\Iv{I}\leq 0$, where
  $\Iv{I}$ is an interval,  and 
  every inequality (one for each $i\in\{1,\dots,n\}$) that does contain $x_1,\dots,x_r$
  to an inequality of the form \[\Iv{p}_1 x_1 + \dots + \Iv{p}_r x_r\leq
  \Iv{q},\] where the $\Iv{p}_1,\dots, \Iv{p}_r, \Iv{q}$ are
  intervals.
\item Replace any inequality of the form $\Iv{I}\leq 0$ with a (not necessarily strictly) negative upper bound of $\Iv{I}$
 by the Boolean constant $\True$ (for ''true'').
\item\label{item:remove} Replace any inequality of the form $\Iv{I}\leq 0$ with a (strictly) positive lower bound of $\Iv{I}$
  by the Boolean constant $\False$ (for ''false'').
\item\label{item:simplify} Simplify the constraint further using basic reasoning
  with the Boolean constants $\True$ and $\False$ (e.g., simplify
  $\True\vee \phi$ to $\True$,
  $\False\vee\phi$ to $\phi$).
\item If the resulting constraint is a Boolean constant, we are done (if the result is the Boolean
  constant $\True$, every $x_1,\dots,x_r$ is a solution, if the result is
  $\False$, no solution exists).
\item\label{item:ivsys} If the constraint is an interval linear system of
  inequalities $\Iv{P}x\leq \Iv{q}$ (i.e., all disjunctions $\vee$ in the formula have
  been removed by the simplifications in Step~\ref{item:simplify}), we reduce
  the interval linear system to a linear system $Az\leq b$ using the method of Rohn and
  Kreslov{\'a}~\cite{Rohn:94,Hladik:13}, and solve this system using linear programming. If the result is not yet
  an interval linear system, we continue with the next step.
\item\label{item:split} If the previous step resulted in a solvable linear program, we have a
  solution to the original problem. If no, we choose an
  $i\in\{1,\dots,n\}$, split the box $\Iv{B}_i$ into pieces $\Iv{B}^1,\Iv{B}^2$, replace
  the original constraint by
\[
\begin{array}{c}
\bigwedge_{j\in \{1,\dots, i-1,i+1,\dots,n\}} [\forall y_1,\dots, y_s\!\in\!\Iv{B}_j\;\; \phi_j(x_1,\dots,x_r, y_1,\dots y_s)]\wedge\\
\begin{array}{l}
 [\forall y_1,\dots, y_s\!\in\! \Iv{B}^1\;\; \phi_i(x_1,\dots,x_r, y_1,\dots
 y_s)]\wedge\\ 
{}[\forall y_1,\dots, y_s\!\in\! \Iv{B}^2\;\; \phi_i(x_1,\dots,x_r, y_1,\dots y_s)]
\end{array}
\end{array}
\] and iterate from Step 1 of the algorithm.
\end{enumerate}

Note that this algorithm only splits boxes with bounds pertaining to the variables $y_1,\dots,y_s$ but
\emph{not} wrt. variables $x_1,\dots,x_r$. This is the main advantage of such an
algorithm over a naive algorithm that substitutes sample points for the free variables
$x_1,\dots,x_n$. Completeness can be preserved due to completeness of the
Rohn/Kreslov{\'a} algorithm which we now describe in more detail:

The basic idea is, to replace each variable $x_i$ by two non-negative variables
$x_i^1$ and $x_i^2$, and then to rewrite each term $[\underline{p_i}, \overline{p_i}]
x_i$ of the interval system of inequalities to $[\underline{p_i}, \overline{p_i}]
(x^1_i-x^2_i)$ which is equal to $[\underline{p_i}, \overline{p_i}] x^1_i- [\underline{p_i},
\overline{p_i}] x^2_i$. Based on the fact that the inequalities should
hold \emph{for all} elements of the intervals, we can now exploit the fact that
the $x_i^1$ and $x_i^2$ are non-negative. Hence, we can now replace the interval
coefficients of $x_i^1$ with their upper endpoint, and the interval coefficients
of $x_i^2$ with their lower endpoint, resulting in  $\overline{p_i} x_i^1 - \underline{p_i} x_i^2$. The
result is an linear system of inequalities of the form $\overline{P}
x^1-\underline{P} x^2 \leq q$, $x_1\geq 0$, $x_2\geq 0$. 

\section{An Informed Splitting Strategy}
\label{sec:heuristics}

The major building block of the algorithm that we left open is the splitting strategy: Which box to choose for splitting in Step~\ref{item:split} and along which variable to split it. In this section we will develop such a strategy. We will first describe the basic idea for the linear system $Az\leq b$ created in Step~\ref{item:ivsys} of the algorithm (sub-section~\ref{sec:basic-idea}), then we will study how to take into account the fact that the linear system was created from an interval linear system $\Iv{P}x\leq \Iv{q}$ by the algorithm of Rohn and Kreslov{\'a} (sub-section~\ref{sec:exploiting-structure}), then we will study how to ensure convergence of the strategy (sub-section~\ref{sec:ensuring-convergence}), take into account interval evaluation from Step~\ref{item:iveval} of the main algorithm (sub-section~\ref{sec:interval-evaluation}), and summarize the result into a sub-algorithm of our main algorithm (sub-section~\ref{sec:algorithm}).

\subsection{Basic Idea}
\label{sec:basic-idea}

 Our goal is to have a strategy that is 
\begin{itemize}
\item complete: if the problem has a solution, we will eventually find
  it\footnote{with the exception of degenerate cases, see Section~\ref{sec:convergence}}
\item efficient: the algorithm converges to a solution as fast as possible.
\end{itemize}

It is not too difficult to ensure completeness of the algorithm: Just ensure
that the width of all boxes goes to
zero~\cite{Ratschan:02f,Ratschan:10}. However, the result can be highly
inefficient: each split increases the size of the constraint to solve, slowing
down the algorithm. Hence it is essential to concentrate on splits that bring
the constraint closer to solvability. 

Since splitting heuristics for classical interval branch-and-bound (or
branch-and-prune) algorithms are
well-studied~\cite{Csendes:97,Douillard:08,Ratschan:02c} we assume that in
Step~\ref{item:ivsys} the algorithm already arrived at an interval linear system
of inequalities. We will try to come up with splits that bring the next linear program
closer to solvability. For achieving this, we need some measure of what it means
for an infeasible system of inequalities $Az\leq b$ (as created by
Step~\ref{item:ivsys} of the algorithm) to be close to solvability, which will lead us to a method for
determining how it can be brought closer to solvability by splitting.

The overall approach is to
\begin{enumerate}
\item use the minimum of the residual $\max_{i\in\{1,\dots,n\}} \;(Az-b)_i$, that is \[\min_{z} \max_{i\in\{1,\dots,n\}}
\;(Az-b)_i\]  as a measure of closeness to feasability (here the index $i$ denotes the $i$-th entry of the vector $Az-b$),
\item to compute the corresponding minimizer, and then to
\item use those splits that promise to improve the residual for this minimizer the most.
\end{enumerate}

For computing the minimum of the residual, we reformulate \[\min_{z} \max_{i\in\{1,\dots,n\}}
\;(Az-b)_i\] as the constrained optimization problem
 \[ \Opt{\min_{z,\rho} \rho}{\rho=\max_{i\in \{1,\dots,n\}} \;(Az-b)_i} \]
which is
 \[ \Opt{\min_{z,\rho} \rho}{\rho\geq (Az-b)_1,\dots, \rho\geq (Az-b)_n}, \]
from where we arrive at the linear program
\[ \Opt{\min_{z,\rho} \rho}{Az-b\leq [1,\dots,1]^T\rho}. \]

Let $z^*, \rho^*$ be the resulting minimizer. If the residual $\rho^*\leq 0$ then we know
that the system $Az\leq b$ is solvable. If not, then the constraint violation vector $Az^*-b$ provides information on
how much the individual constraints contribute to non-solvability.

We try to decrease the constraint violation of the row of $A$ for which $Az^*-b$
is maximal. Denote this row by $i$. The constraint corresponding to this row is of the form $a_1 z^*_1+\dots+a_{2r} z^*_{2r}\leq b$ where each coefficient $a_j, j\in\{1,\dots,2r\}$ results from an endpoint of some interval in the interval system $\Iv{P}x\leq \Iv{q}$. Now we want to choose a $j\in \{1,\dots, 2r\}$ such that splitting will aim at changing the coefficient $a_j$ as much as possible. We assume that the change that we can expect for coefficient $a_j$ if using such a split, is given by some real number $\delta_j$. Under this assumption, the inequality will change to \[a_1 z^*_1+\dots+ a_{j-1} z^*_{j-1}+ (a_j +\delta_j) z^*_j + a_{j+1} z^*_{j+1}+\dots+a_{2r} z^*_{2r}\leq b\] which is \[a_1 z^*_1+\dots+a_{2r} z^*_{2r}\leq b-\delta_j z_j^*,\] resulting in an improvement $-\delta_j z_j^*$. 

Hence we can expect the maximal improvement of the residual by choosing $j$ as 
\[ \arg \max_{j\in\{1,\dots,2r\}} -\delta_j z^*_j \]

For analyzing how $\delta_j$ should look like, we have to analyze the precise form of the system $Az-b$ which we will do in the next sub-section.


\subsection{Exploiting Structure}
\label{sec:exploiting-structure}
 

Now observe that the linear program that we used in the previous sub-section is not arbitrary, but is the
result of the Rohn/Kreslov{\'a} transformation of an interval system of linear
inequalities of the form \[\Opt{\min \rho}{\overline{P}x^1-\underline{P}x^2  
 -b \leq [1,\dots,1]^T\rho, x^1\geq 0, x^2\geq 0}.\]

Observe that the entries of the underlying interval linear system of inequalities $\Iv{P}x\leq \Iv{q}$ are created by interval evaluation (Step~\ref{item:iveval} of the main algorithm). Assuming that splitting shrinks large entries of the interval matrix $\Iv{P}$ more than small intervals, the change $\delta_j$ that we can expect for $a_j$ from
splitting is proportional to the width $\IWidth{\Iv{p}_{v(j)}}$ of the
corresponding interval $\Iv{p}_{v(j)}$ in the $i$-th row of
$(\Iv{p}_1,\dots,\Iv{p}_n)$ of $\Iv{P}$. However, since splitting results in a sub-interval $\Iv{p}_{v(j)}'\subseteq \Iv{p}_{v(j)}$, the expected change for lower bounds of intervals is positive, and for upper bounds of intervals is negative. 

Analyzing the left-hand side $\overline{P}x^1-\underline{P}x^2 -b$ of 
the linear program resulting from the Rohn/Kreslov{\'a} transformation, we observe that the $x^1$ have coefficient $\overline{P}$, that is, the sign of upper bounds is positive, and the $x^2$ have coefficient $-\underline{P}x^2$ that is, the sign of lower bounds is negative. Combining this with the fact that lower bounds will be increased and upper bounds be decreased by splitting, the expected change $\delta_j= -\IWidth{\Iv{p}_{v(j)}}$, resulting in \[ \arg \max_{j\in\{1,\dots,2r\}} -\delta_j z_j = \arg \max_{j\in\{1,\dots,2r\}} \IWidth{\Iv{p}_{v(j)}} z_j^*.\]

Now observe furthermore, that the coefficients of the linear program come in pairs that
refer to the two bounds the same intervals, and hence also their width is the same.
So, instead of \[ \arg \max_{j\in\{1,\dots,2r\}} \IWidth{\Iv{p}_{v(j)}} z_j^* \]
we can directly refer to the interval matrix:
\[ \arg \max_{j\in\{1,\dots,r\}} [\IWidth{\Iv{p}_j} \max\{ x^1_j, x^2_j\} ]\]
where $x^1_j$ and $x^2_j$ refer to the individual entries of the vectors of variables as introduced by the
Rohn/Kreslov{\'a} transformation. 

 We also note the following:
\begin{lemma}
\label{lem:4}
Let $\overline{P}$ and $\underline{P}$ be real matrices in $\RR^{n\times r}$ such that for every $i\in \{1,\dots,n\}$, $j\in \{1,\dots r\}$, $\underline{P}_{i,j}<\overline{P}_{i,j}$. Let $b$ a real vector in $\RR^n$. 
Then for every solution $x^1$, $x^2$ of the linear program \[\Opt{\min \rho}{\overline{P}x^1-\underline{P}x^2  
 -b \leq [1,\dots,1]^T\rho, x^1\geq 0, x^2\geq 0}\]
for every $j\in \{1,\dots, r\}$, either $x^1_j$ or $x^2_j$ is zero.
\end{lemma}

\begin{proof}
Let $\Iv{P}^c$ be the center of the interval matrix $\Iv{P}$, that is the matrix
that contains the midpoint of the corresponding intervals of $\Iv{P}$. Let
$\Iv{P}^\triangle$ be the matrix that contains for every entry the width of the
corresponding interval of $\Iv{P}$. Then the above linear program is equivalent
to \[ \min \rho, \Iv{P}^c(x^1-x^2) + \Iv{P}^\triangle (x^1+x^2)
 -b \leq [1,\dots,1]^T\rho, x^1\geq 0, x^2\geq 0 \]

Let $j\in \{1,\dots, r\}$ be arbitrary, but fixed, and assume that both $x^1_j$
and $x^2_j$ are non-zero. Then, we can replace $x^1_j$ by $x^1_j-\varepsilon$ and
$x^2_j$ by $x^2_j-\varepsilon$, where $\varepsilon>0$. As a result, the value of
the first term $\Iv{P}^c(x^1-x^2)$ stays unchanged, while the value of the
second term $\Iv{P}^\triangle (x^1+x^2)$ has decreased. Hence we can decrease
the minimum of the linear program, which is a contradiction to the assumption
that the original values $x^1_j$ or $x^2_j$ were a solution of the linear program.
 \QED
\end{proof}

\subsection{Ensuring Convergence}
\label{sec:ensuring-convergence}

The basic idea, as described in the previous section, does not result in a
converging method. We will demonstrate this on a concrete example, taking into
account the precise form of how the system of inequalities is created by the
method of Rohn and Kreslov{\'a}. For this, assume the interval
inequality \[ [-1,3]x_1+[-3,1]x_2\leq -2\] and the corresponding inequality \[ 3x_1^1
+x_1^2 +x_2^1 +3x_2^2\leq -2, x_1^1\geq 0, x_1^2\geq 0, x_2^1\geq 0, x_2^2\geq 0.\] The
resulting linear program \[ \Opt{\min_{x_1^1, x_1^2, x_2^1, x_2^2, \rho} \rho}{3x_1^1 +x_1^2 +x_2^1 +3x_2^2 +2 
\leq[1,\dots,1]^T\rho, x_1^1\geq 0, x_1^2\geq 0, x_2^1\geq 0, x_2^2\geq 0}\] has the solution
$\rho=2$, $x_1^1=0$, $x_1^2=0$, $x_2^1=0$, $x_2^2=0$ which corresponds to the
values $x_1=0$, $x_2=0$ of the original interval inequality. Evaluating our heuristics, we get 
\[ \arg \max_{j\in\{1,2\}} [\IWidth{\Iv{p}_j} \max \{ x^1_j, x^2_j\} ]=
 \arg\max_{j\in\{1,2\}} 0=0 \]

Hence our heuristics already compute the value $0$ for each coefficient, suggesting that no shrinking of interval coefficients is necessary
anymore (Theorem~\ref{thm:1} in Section~\ref{sec:convergence} will provide a more general characterization of such behavior). Still, we have not yet found a solution of $[-1,3]x+[-3,1]y\leq -2$, and the residual value $\rho>0$ correctly indicates this. Moreover, a shrinking of the first interval, for example, resulting in \[ [2,3]x+[-3,1]y\leq -2\] leads to a solvable system.

Analyzing the problem, we see that for the solution of $[2,3]x+[-3,1]y\leq -2$, $x\neq 0$! So the original
heuristics was misleading, since it mistakenly assumed $x=0$. In other words,
while the minimizer $z^*$ of the linear program $\min_{z,\rho} \rho, Az-b\leq
[1,\dots,1]^T\rho$ gives some orientation on which coefficients to
shrink, it need not necessarily be a solution of the original input constraint,
and hence may be misleading.

To fix the problem, we assume that the minimizer $x^1, x^2$ to the linear program only approximates the final solution of the input constraint that we are looking for. For each $j\in\{1,\dots, r\}$, the final solution might instead be located in an interval
$[x_j^1-x_j^2-\varepsilon, x_j^1-x_j^2+\varepsilon]$ around the corresponding solution $x_j^1-x_j^2$ of the interval system of linear
inequalities.

We will now analyze the corresponding changed value of the term $\max\{ x^1_j, x^2_j\}$ used in the computation of the heuristic value 
\[\arg \max_{j\in\{1,\dots,r\}} [\IWidth{\Iv{p}_j} \max\{ x^1_j, x^2_j\} ].\] 
In the case where $x_j^1-x_j^2\geq 0$, by Lemma~\ref{lem:4}, $x_j^2=0$. Hence the original value of the term $\max\{ x^1_j, x^2_j\}$ is $x^1_j$, and the corresponding changed value is $x_j^1+\varepsilon$. In the case where $x_j^1-x_j^2< 0$, the original value is $x_j^2$, and the changed value is $x_j^2+\varepsilon$. Putting those cases together, the changed value is $\max \{  x^1_j, x^2_j\}+\varepsilon$.


Hence the corresponding heuristic value can be up to 
\[ \arg \max_{j\in\{1,\dots,r\}} \big[\IWidth{\Iv{p}_j} [\max \{  x^1_j, x^2_j\}+\varepsilon]\big] \]
which we will use, for a user-provided constant $\varepsilon>0$. 

We will see later  (Section~\ref{sec:convergence}), that even if the constant $\varepsilon$ does not correctly estimate the difference between the  solution of the current linear program and a solution of the original constraint $\phi$, for constraints that have a non-degenerate solution, the resulting
method always converges to such a solution, if $\varepsilon>0$.

In general, we will use heuristics of the form
\[ \arg\max_{j\in\{1,\dots,r\}} h(\Iv{p}_j, x^1_j, x^2_j)\]
and show convergence under certain conditions of this function $h$.

\subsection{Interval Evaluation}
\label{sec:interval-evaluation}

Up to now, we know which row(s) of $\Iv{P}$ to split, that is, for which $i\in\{
1,\dots, n\}$ to split the box $\Iv{B}_i$. We also know, which bound of which interval
in that row of $\Iv{P}$ we want to decrease, but we still do not know which coordinate of $\Iv{B}_i$
result in the biggest decrease of that bound. 
For determining this, observe that each entry of the interval matrix $\Iv{P}$ results from interval
evaluation of a certain expression on the box $\Iv{B}_i$. Hence we need to infer, for
a given arithmetical expression and an interval for each variable in that
expression, which split of an interval results in the biggest decrease of the
given resulting (lower or upper) bound of interval evaluation. There are many possible choices for this. Hence our
approach will be parametric in the concrete method used. We will assume a
function $\text{splitheur}$ such that for a given arithmetical expression $t$,
box $\Iv{B}$, and sign $s\in \{ -, +\}$, $\text{splitheur}(t, \Iv{B}, s)$ 
\begin{itemize}
\item returns a variable of $\Iv{B}$ to split for improving the lower/upper bound
  (depending on $s\in \{ -, +\}$) of
  the interval evaluation of $t$ on $\Iv{B}$, and for which
\item repeated splitting according to this function converges, that is, for
  the sequence $\Iv{B}^1,\Iv{B}^2,\dots$ created by splitting according to this function, for
  every $\varepsilon>0$ there is a $k$ such that for all $i\geq k$ the width
  of $t(\Iv{B}^i)$ is smaller than $\varepsilon$.
\end{itemize}

Right now, we use this function for only the coefficient chosen by our heuristics. It might also
make sense to try it on all coefficients, and choose the best one.

\subsection{Algorithm}
\label{sec:algorithm}

The resulting algorithm is called from the main algorithm in
Step~\ref{item:split}  in the case where in
Step~\ref{item:ivsys} we arrived at an (unsolvable) interval linear system. The
algorithm has the following form (where, for an arithmetical expression $t$ and a  box
$\Iv{B}$, $\IEval{t}{\Iv{B}}$ denotes interval evaluation of $t$ on $\Iv{B}$):

\begin{ntabbing}
\textbf{Input: } \=
expressions $t_{i,j}, i\in\{ 1,\dots, n\}, j\in \{ 1,\dots,
r\}$ s.t. \Style{}{\\\`} $t_{i,j}$ is the coefficient of $x_j$ in $\phi_i$,\\
\> boxes $\Iv{B}_i, i\in \{1,\dots,n\}$\\
\textbf{Output: } $i\in\{1,\dots,n\}, k\in\{1,\dots, s\}$
 \Style{}{\\\`} suggesting to split box $\Iv{B}_i$ at its $k$-th coordinate\\
\textbf{let} $\Iv{P}$ be the $(n\times r)$-interval matrix s.t. $\Iv{P}_{i,j}= \IEval{t_{i,j}}{\Iv{B}_i}$, $i\in\{ 1,\dots, n\}, j\in \{ 1,\dots,
r\}$\label{}\\
$(x^1, x^2)\leftarrow \arg\min \rho$, $\overline{P}x^1-\underline{P}x^2  -b 
\leq [1,\dots,1]^T\rho, x^1\geq 0, x^2\geq 0$ \label{l:solve_relax}\\
$d\leftarrow \overline{P} x^1 - \underline{P} x^2- b$ \label{l:residuum}\`// \textit{residual}\\
$i\leftarrow \arg\max_{i\in \{1,\dots n\}} d_i$ \label{l:max_viol}\`// \textit{box $\Iv{B}_i$ to split}\\
$j\leftarrow\arg\max_{j\in\{1,\dots,r\}} h(\Iv{P}_{i,j}, x^1_j, x^2_j)$\label{l:choose_column}\`// \textit{coefficient to improve}\\
\textbf{return} $i$, $\text{splitheur}(t_{i,j}, \Iv{B}_i, \Sign{x^1-x^2}_j)$ \label{}
\end{ntabbing}

We will call this version of the algorithm the split-worst version. We will also
consider an alternative version that, instead of splitting only the box
corresponding to the maximal constraint violation (as computed in
Line~\ref{l:max_viol}), splits \emph{all} boxes with positive constraint
violation. We will call that version of the algorithm, the split-all version.

As already discussed above, if for the solution $\rho$ computed in
Line~\ref{l:solve_relax}, $\rho\leq 0$, then we know that the interval linear
system $\Iv{P}x\leq \Iv{q}$ from Line~\ref{item:ivsys} of the main algorithm is
solvable. Moreover, since $\rho\leq 0$ is equivalent to the linear system of
Rohn/Kreslov{\'a} being solvable, this computes the same information as
Line~\ref{item:ivsys} of the main algorithm and hence no solving has to be done there. In other
words, instead of solving the Rohn/Kreslov{\'a} linear system of equations, we solve
the linear program \[\min \rho, \overline{P}x^1-\underline{P}x^2  
 -b \leq [1,\dots,1]^T\rho, x^1\geq 0, x^2\geq 0\]
and use it both for determining the overall solution of the algorithm and
heuristics for splitting.

\section{Equality Constraints}
\label{sec:equality-constraints}

Now we analyze the extended problem that---in
addition to constraints of the form $\forall y_1,\dots, y_s\!\in\!
\Iv{B}_i\;\; \phi_i(x_1,\dots,x_r, y_1,\dots y_s)$---also contain linear equalities
over the variables $x_1,\dots, x_r$. 
Viewing each equality as a conjunction of two inequalities one sees that in that
case, the two inequalities force the optimum $\rho^*$ of
\[\min \rho, \overline{P}x^1-\underline{P}x^2  -b\leq
[1,\dots,1]^T\rho, x^1\geq 0, x^2\geq 0\] to be zero. So in this
case, the heuristics in the form described above are not useful. 
In order to handle equalities better, we do not view such equalities as two
inequalities, but we handle them directly. That is we solve the linear program 

\[\min \rho, \overline{P}x^1-\underline{P}x^2  -b\leq
[1,\dots,1]^T\rho, C(x^1-x^2)= d, x^1\geq 0, x^2\geq 0\]
where  $Cx=d$ is the linear system of
equations containing all the  linear equalities.

\section{Convergence}
\label{sec:convergence}

Our main algorithm consists of a loop that continues until a solution has been
found. In this section we will answer the question: Will the loop terminate for
all input constraints? Again we will assume that 
in Step~\ref{item:ivsys} the algorithm already arrived at an
interval linear system of inequalities, since convergence of basic interval
branch-and-bound (or branch-and-prune) algorithms is not difficult to show
(e.g., it follows as a special case of Theorem~6 in~\cite{Ratschan:02f}).

Observe that the only place where the algorithm approximates, is the interval
evaluation in Step~\ref{item:iveval} of the main algorithm. In the whole
section, for the formal proofs, we assume that the resulting linear programs are
solved precisely, using rational number arithmetic. Still, in practice, it
suffices to solve them approximately, for example, based on floating-point
arithmetic.

In the following we will denote by $\Iv{\hat{P}}x\leq \Iv{\hat{q}}$ the
system of linear interval inequalities that would result from the input
constraint if interval evaluation would be non-overapproximating. In a similar
way, we will denote by $\hat{A} z\leq \hat{b}$ the  system of linear
inequalities corresponding to $\Iv{\hat{P}}x\leq \Iv{\hat{q}}$. 

\begin{definition}
  We call $z$ a \emph{robust solution}  of a system of inequalities and equalities
  $Az\leq b \wedge Cz=d$ iff $Az<b\wedge Cz=d$. We call a constraint $\phi$ \emph{robust} if the
  corresponding system $\hat{A} z\leq \hat{b}\wedge Cz=d$ has a robust solution.
\end{definition}

In other words, a robust solution of a system $Az\leq b \wedge Cz=d$ is an
interior point of $Az\leq b$ that satisfies the equalities $Cz=d$.

\begin{lemma}
\label{lem:3}
  If $z$ is a robust solution of   $Az\leq b \wedge Cz=d$, then there is an
  $\varepsilon>0$ s.t. that for all $A'$, and $b'$ differing from $A$ and $b$
  not more than $\varepsilon$ for each entry, $A'z\leq b' \wedge Cz=d$.
\end{lemma}

\begin{proof}
  Since $z$ is a robust solution, $Az-b$ is a vector of negative numbers. From
  this we can compute an upper bound on the allowed changes of $A$ and $b$. \QED
\end{proof}

It is not difficult to ensure convergence of the algorithm:

\begin{lemma}
\label{lem:2}
  Assume that the splitting strategy ensures that the width of every bounding
  box $\Iv{B}_i$
  goes to zero. Then the algorithm will terminate for robust inputs.
\end{lemma}

\begin{proof} Assume an arbitrary iteraton of the algorithm. As above, denote by $\Iv{\hat{P}}x\leq \Iv{\hat{q}}$ the interval system of inequalities that the algorithm would compute if using the precise range instead of over-approximating interval evaluation in Step~\ref{item:iveval}. Denote by $\hat{A} z\leq \hat{b}$ the corresponding system of linear inequalities the algorithm computes from $\Iv{\hat{P}}x\leq \Iv{\hat{q}}$ in Step~\ref{item:ivsys}. Assuming, in addition, a system of linear equalities $Cz=d$, 
let $\hat{z}$ be the robust solution of $\hat{A} z\leq
  \hat{b}\wedge Cz=d$, so $\hat{A}
  \hat{z} < \hat{b}$. Due to the fact, that the algorithm does not compute the precise range in Step~\ref{item:iveval}, but over-approximates it using interval evaluation, the algorithm will compute
with an interval matrix $\Iv{P}\supseteq \Iv{\hat{P}}$. The over-approximation error goes to
zero due to convergence of interval arithmetic. Hence $\hat{A}$ will be approximated increasingly
well. So, due to Lemma~\ref{lem:3}, $\hat{A} z\leq \hat{b}\wedge Cz=d$ will eventually hold and the algorithm terminates. \QED
\end{proof}

However, our heuristics do not necessarily ensure that the width of every
bounding box goes to zero: Even if it would ensure that every bounding box is
split infinitely often, it might still happen that the width of some bounding
box does not go to zero, because a certain coordinate of the box is not split
infinitely often. This might not even be necessary for termination, because this
coordinate might correspond to a variable that does not occur in an coefficient
term.

\begin{theorem}
\label{thm:1}
Consider the split-all version of the algorithm with heuristics of the form
\[ \max_{j\in\{1,\dots,n\}} h(\Iv{p}_{j}, x^1_j,x^2_j)\]
where 
\begin{enumerate}[(a)]
\item\label{ass:conv} $\lim_{\IWidth{p}\rightarrow 0 }h(p, x)=0$
\item\label{ass:nonzero} $\IWidth{p}>0$ implies $h(p, x)>0$.
\end{enumerate}
Then we have: If the input constraint has a robust
solution, then the algorithm terminates.
\end{theorem}

\begin{proof}
  We consider the split-all version of the algorithm and assume that the algorithm does not
  terminate. Then it creates an infinite sequence of unsolvable interval linear programs 
  and corresponding linear programs (see Line~\ref{item:ivsys} of the main algorithm). In
  each iteration all $\Iv{B}_i$ with positive constraint violation are
  split. Hence, all those bounding boxes are split infinitely often. In each
  iteration, in Line~\ref{l:choose_column} of the algorithm from
  Section~\ref{sec:algorithm}, a coefficient $j$ for improvement is chosen.
  If all coefficients are chosen infinitely often, then due to
  convergence of splitheur the width of all coefficients goes to zero, which
  implies that the constraint will eventually have non-positive constraint
  violation and the corresponding bounding box would not be chosen for
  splitting, which is a contradiction.

  Hence, non-termination implies that at least one of the coefficients is not chosen infinitely
  often. Let us analyze the state of the algorithm where all coefficients that are
  chosen finitely often will not be chosen any more. All other coefficients are
  chosen infinitely often, which means that due to convergence of splitheur,
  their interval width goes to zero. Hence, due to Assumption~(\ref{ass:conv}),
  their $h$-value goes to zero. Moreover, due to Assumption~(\ref{ass:nonzero})
  the coefficients that are not split any more, have positive $h$-value. This
  implies that  Line~\ref{l:choose_column} eventually chooses one of them, a
  contradiction. So the algorithm terminates. 

\QED
\end{proof}

Clearly, the heuristics\[ \arg \max_{j\in\{1,\dots,r\}} \big[\IWidth{\Iv{p}_j}
[\max \{  x^1_j, x^2_j\}+\varepsilon]\big], \]  with $\varepsilon>0$, as developed in
Section~\ref{sec:heuristics}, fulfill the assumptions of the theorem, and hence
the algorithm converges.

\section{Variable Splitting Heuristics}
\label{sec:variable_splitting}

In this section we discuss, how the function $\text{splitheur}(t, \Iv{B}, s)$, that
we 
introduced in Section~\ref{sec:interval-evaluation}, can be implemented.

A widely used technique (e.g., in global optimization~\cite{Csendes:97}) for
this is to use derivatives of the arithmetical expression. An alternative would
be Corollary 2.1.2 in Neumaier's book~\cite{Neumaier:90}. However, that would
need the computation of interval over-approximation of derivatives, or interval
Lipschitz constants, respectively. Moreover, those
techniques are only a priori estimates of the decrease that might fail to give
exact information. In order to arrive at more precise information, we use the
observation that we already have a fixed set of usually small expressions that we
want to analyze. Hence, interval evaluation of those expressions will usually take
negligible time compared to the rest of the algorithm. Hence we explicitly try
all possible splits and compare their effect on the width of the result of
interval evaluation~\cite{Csendes:01}: 

\[
\mbox{splitheur}(t, \Iv{B}, s)= \arg\max_{i\in \{ 1,\dots, |\Iv{B}|\}} \min_k \{
|b^s(t(\Iv{B}))-b^s(t(\mbox{split}^k(\Iv{B}, i)))| \}  
\]

where $b^s$ takes the upper-/lower bound respectively of the argument interval
according to $s$, and $\mbox{split}^k(\Iv{B}, i)$ denotes the $k$-th box resulting
from splitting the box $\Iv{B}$ at variable $i$ (usually $k\in\{1,2\}$).

However, the method as described up to now does not ensure convergence of the method. The
reason is the following: Assume a term $t$ in $n$ variables. Assume intervals
$\Iv{I}_1,\dots, \Iv{I}_n$ on which we evaluate $t$. Let $\Iv{I}_i^-$ be the lower and $\Iv{I}_i^+$
be the upper half of $\Iv{I}_i$. Assume a procedure that replaces that interval $\Iv{I}_i$
by its lower or upper half, for which this results in the biggest decrease of
interval evaluation of $t$. Repeated application of this procedure does not
result in the width of interval evaluation going to zero. For example, for the
term $x^2+y$ with $x\in[-1,1]$, $y\in[-0,2]$, splitting $[-1,1]$ does not result
in any improvement at all. However, it is necessary for global convergence.

One way of solving this problem is, to take the time since the last split into
account. For example, we could use
\begin{multline*}
\mbox{splitheur}(t, \Iv{B}, s)=\\ \arg\max_{i\in \{ 1,\dots, |\Iv{B}|\}} \big[c(i)+\min_k \{
|b^s(t(\Iv{B}))-b^s(t(\mbox{split}^k(\Iv{B}, i)))|\}  \big] 
\end{multline*}

where $c(i)$ is a function that increases with the time of the last split of
variable $i$. If this function goes to infinity with the time of the last split,
then every variable will be split eventually, ensuring convergence. The result
is some compromise between round-robin-splitting (which ensures convergence) and
aggressive local improvement. In order to make this heuristics independent of
the size of $\Iv{B}$ (which decreases during the algorithm) it makes sense to use
some scaling with $\IWidth{t(\Iv{B})}$  in the function $c(i)$.

\Long{Instead of analyzing the result of splitting, analyze the common face of the two
splits? }




\section{Computational Experiments}
\label{sec:experiments}

We did experiments on examples for computing Lyapunov-like
functions~\cite{Ratschan:10}, with the heuristic function 
\[ \arg \max_{j\in\{1,\dots,r\}} \big[\IWidth{\Iv{p}_j} [\max \{  x^1_j, x^2_j\}+\varepsilon]\big] \]
and $\varepsilon=0.001$.

For the resulting examples we have $\phi_1=\dots=\phi_r$ with different bounding
boxes for each branch $i\in\{1,\dots,r\}$. The bounding boxes and the inequality
constraints of the examples are as follows:

Example A: 
\[
\begin{array}{ll}
\Iv{B}_1=[0.8, 1.2] \times [0.3,0.49],& \Iv{B}_2= [0.8, 1.2] \times [0.51,0.7],\\
\Iv{B}_3=[1.01, 1.2]\times [0.49,0.51],& \Iv{B}_4=[0.8, 0.99]\times [0.49,0.51]
\end{array}
\]

where $\phi$ is of the form
\begin{multline*}
x_1(2 y_1^3 y_2 -2 y_1^2 + y_1)+x_2 (y_1^2 y_2- y_1 +0.5)+x_3(y_1^2 y_2^2- y_1^3y_2- y_1 y_2 +0.5 y_1 +0.5 y_2)+\\ x_4 (0.5- y_1 y_2^2)+x_5((-2) y_1^2 y_2^2 + y_2)\leq-0.0001
\end{multline*}

Example B: 
\[
\begin{array}{ll}
\Iv{B}_1= [-0.8, 0.8]\times [-0.8,-0.1], &
\Iv{B}_2= [-0.8, 0.8]\times [0.1,0.8], \\
\Iv{B}_3= [-0.8, -0.1]\times [-0.1,0.1], &
\Iv{B}_4= [0.1, 0.8], [-0.1,0.1]
\end{array}\]

where $\phi$ is of the form
\[
x_1 (-2 y_1^2+2 y_1 y_2)+ x_2 (0.2 y_1 y_2-4 y_2^2-2 y_1^2 y_2-0.2 y_1^3 y_2)\leq -0.0001
\]
 
Example C: 
\[
\begin{array}{ll}
\Iv{B}_1=[-0.4, 0.4]\times [-0.4,-0.1], &
\Iv{B}_2=[-0.4, 0.4]\times [0.1,0.4], \\
\Iv{B}_3=[-0.4, -0.1]\times [-0.1,0.1], &
\Iv{B}_4=[0.1, 0.4]\times [-0.1,0.1] 
\end{array}\]

where $\phi$ is of the form
\begin{multline*}
x_1 (-16 y_1^6+24 y_1^5-8 y_1^4)+x_2 (-12 y_1^5+18 y_1^4-6 y_1^3)+\\x_3 (-8 y_1^4+12 y_1^3-4 y_1^2)+x_4 (-4 y_2^2)\leq
-0.000001
\end{multline*}

Example D: 

\[
\begin{array}{ll}
\Iv{B}_1=[-0.2, 0.2]\times [-0.2,0.2]\times [-0.2,-0.1],&
\Iv{B}_2=[-0.2, 0.2]\times [-0.2,0.2]\times [0.1,0.2], \\
\Iv{B}_3=[-0.2, 0.2]\times[-0.2,-0.1]\times [-0.1,0.1]&
\Iv{B}_4=[-0.2, 0.2]\times [0.1,0.2]\times [-0.1,0.1],\\
\Iv{B}_5=[-0.2, -0.1]\times [-0.1,0.1]\times [-0.1,0.1],&
\Iv{B}_6=[0.1, 0.2]\times [-0.1,0.1]\times [-0.1,0.1]
\end{array}\]

where $\phi$ is of the form
\begin{multline*}
x_1 (-2 y_1 y_2)+x_2 (-2 y_2 y_3)+x_3 (-2 y_3^2-2 y_1 y_3+ 2 y_1^3 y_3)+x_4 (- y_1^2 - y_1 y_3)+\\x_5 (y_1^2-2 y_1
y_2 - y_1 y_3 - y_2 y_3+y_1^4)+x_6 (-2 y^2-y_3^2-y_1 y_2-y_2 y_3 +y_1^3 y_2)\leq -0.0001  
\end{multline*}

In all four cases, we normalized the first coefficient $a$ to $1$. To create
versions with equality constraints we used the pre-processing method described 
in Section~4.3. of ~\cite{Ratschan:10}. We will denote the result by $A'$, $B'$,
$C'$ and $D'$.

The results of the experiments can be seen in Figure~\ref{fig:results}. Here, round-robin refer to
the classical round-robin splitting heuristics where variables are split one
after the other, and that we used in earlier work~\cite{Ratschan:10}. 
Empty entries correspond to cases where the algorithm did not terminate within
$10$ minutes, and $\varepsilon$ corresponds to cases where the algorithm
terminates in less than $0.1$ seconds. The experiments were done based on an
implementation in the programming language Objective Caml using the LP solver
Glpk, on a Linux operating system and a 64-bit 2.80GHz processor).

\begin{figure}[htb]
  \centering

\begin{tabular}{|l|r|r|r|r|r|r|}\hline
  & \multicolumn{2}{c|}{round-robin} & \multicolumn{2}{c|}{split-worst} &
  \multicolumn{2}{c|}{split-all} \\\hline
   & splits & time & splits & time & splits & time\\\hline
A & 48 & 0.260 & 21 & $\varepsilon$ & 5 & $\varepsilon$ \\
B & 545 & 8.777 & 112 & 0.648 & 6 & $\varepsilon$\\
C &     &       & 20 & 0.128 & 5 & $\varepsilon$\\
D &    &        & 719 & 70.176 & 10 & 2.472\\
A' & 0 & $\varepsilon$ & 0 & $\varepsilon$ & 0 & $\varepsilon$\\
B' & 719 & 19.233 & 82 & 0.300 & 9 & $\varepsilon$\\
C' & 244 & 6.472 & 4 & $\varepsilon$ & 3 & $\varepsilon$ \\
D' & 0 & $\varepsilon$ & 0 & $\varepsilon$ & 0 & $\varepsilon$ \\\hline
\end{tabular}

  \caption{Results of Experiments}
  \label{fig:results}
\end{figure}



\section{Conclusion}
\label{sec:conclusion}

We have shown how to efficiently solve a class of quantified
constraints. Computational experiments show that the corresponding splitting
heuristics result in efficiency improvements by orders of magnitude. 

Possibilities for further research include:
\begin{itemize}
\item The application of the algorithm to areas such as termination
  analysis~\cite{Cousot:05,Lucas:07,Podelski:04}. 
\item The extension of the algorithm to a more general class of constraints, for
  example to ensure applicability in invariant computation~\cite{Rodriguez-Carbonell:04,Sankaranarayanan:08a,Sturm:11}.
\end{itemize}

\bibliographystyle{plain}
\bibliography{sratscha}

\Long{
\appendix

\section{Further Ideas}

send paper to: Lars Gruene, Fabian Wirth, Peter Giesl etc

Comparison with Rybalchenko's Farkas based approach: We have
interval coefficients, whereas his constraint ensuring that the ranking
function decreases has constant coefficients. Instead of the box constraint
$y\in \Iv{B}$ he has a polyhedral constraint resulting from the conjunction forming
the loop condition. The intersection of both cases (box constraint and constant
coefficients) both methods (probably) behave similarly since in that case
Rohn/Kresolva does not increase number of variables\footnote{in theory, but how
  does our implementation behave?} and solves an LP (with many constraints but
few variables), whereas Rybalchenko solves its dual.

add monotonicity check (we
only have to check boundaries and parts of state space, where derivatives might
be zero), in case of ODEs: add with ODE diagonalization, different than monomial basis?

right now we ignore the fact that in the LP, we have pairs of coefficients
resulting from the same interval, exploit this?

Different vector than $[1,\dots,1]^T$? changing $\varepsilon$ during
computation? non-uniform $\varepsilon$?

heuristics for Boolean structure?

applications:
\begin{itemize}
\item computation of terminal cost function in model predictive control
\item 
Lyapunov functions for distance measure in planning (see LaValle
book). 
\item 
Verification of hybrid systems---relaxation of barrier: function that is
negative on initial positive on unsafe and decreases in between

Any different way of choosing best coefficient to modify? see parametric
programming, post-optimality/sensitivity analysis (LP books available on UI network), Chinneck book,
book on multi-parametric programming; parametric programming in encycl. of optimization

check against linear algebra preprocessing
\end{itemize}
}

\end{document}